\documentclass[letterpaper, 10 pt, conference]{ieeeconf}
\IEEEoverridecommandlockouts
\usepackage{cite}
\usepackage{url}

\usepackage{amsmath,amssymb,amsthm,amsfonts}
\usepackage{mathtools}
\mathtoolsset{showonlyrefs=true}
\usepackage{algorithmic}
\usepackage{graphicx}
\usepackage{textcomp}
\usepackage{xcolor}
\usepackage[caption=false,font=footnotesize]{subfig}
\usepackage{comment}
\usepackage{url}
\newtheorem{remark}{Remark}

\newtheorem{proposition}{Proposition}

\usepackage{algorithm}

\title{\LARGE \bf
Safe Decentralized Density Control of Multi-Robot Systems using PDE-Constrained Optimization with State Constraints
}

\author{Longchen Niu and Gennaro Notomista%
\thanks{The authors are with the Department of Electrical and Computer Engineering, University of Waterloo, Waterloo, ON, Canada {\tt\footnotesize {l3niu,gennaro.notomista}@uwaterloo.ca}}%
\thanks{The authors thank Tahmid Sheikh for assisting with drone configuration and experiment setup, and Kaurisan Selvarasa and Alyson Colpitts for their help during testing at the RoboHub.}
}

\begin{document}
\thispagestyle{empty}
\onecolumn   
© 2025 IEEE.  Personal use of this material is permitted.  Permission from IEEE must be obtained for all other uses, in any current or future media, including reprinting/republishing this material for advertising or promotional purposes, creating new collective works, for resale or redistribution to servers or lists, or reuse of any copyrighted component of this work in other works.

Accepted to: 2025 IEEE International Symposium on Multi-Robot and Multi-Agent Systems (MRS)

\newpage
\twocolumn

\maketitle
\thispagestyle{empty}
\pagestyle{empty}

\begin{abstract}
In this paper, we introduce a decentralized optimization-based density controller designed to enforce set invariance constraints in multi-robot systems. By designing a decentralized control barrier function, we derived sufficient conditions under which local safety constraints guarantee global safety. We account for localization and motion noise explicitly by modeling robots as spatial probability density functions governed by the Fokker-Planck equation. Compared to traditional centralized approaches, our controller requires less computational and communication power, making it more suitable for deployment in situations where perfect communication and localization are impractical.
The controller is validated through simulations and experiments with four quadcopters.
\end{abstract}

\section{Introduction}
In recent years, the development of distributed algorithms, together with the increasing affordability of hardware, has promoted the deployment of teams of robots in the real world at scale. One area of focus with a wide range of applicability is
density control, where a team of robots works together to manipulate their spatial distribution to achieve area coverage goals. This paradigm has many promising applications, including wildfire control, environment monitoring, precision agriculture, and search and rescue \cite{Fire, application_enivronment, application_SaR, zhai2021cooperative}. In these applications, a team of robots provides significant advantages with their scalability and robustness in a changing environment. In many scenarios, however, it can be critical for the safety of the robots to avoid certain areas, or mathematically speaking, to enforce set invariance constraints. Studies on computationally efficient ways of enforcing safety constraints have been conducted on teams of robots modeled by Ordinary Differential Equations (ODEs). Examples include \cite{wang2017safety, CBF, safety_car_distance, Optimal_withInvariance, OptimalMain, GassuainVoronoiPathPlanning}.

A traditional centralized controller, where all robots rely on one central computer for commands, is applicable only to smaller-sized multi-robot systems and situations with perfect communication. 
One popular approach leverages mean-field game theory by modeling robots as Probability Density Functions (PDFs) \cite{Elamvazhuthi_Berman_2019}, and has been recently applied to optimal control with state invariance \cite{Optimal_withInvariance, OptimalMain}. 
While these approaches ensure a globally optimal solution in space and time, their dependence on pre-computed control fields limits performance in uncontrolled outdoor applications with unmodeled disturbances, such as animals or wind. 

For applications with a larger number of robots or larger fields, a decentralized controller is a more suitable choice. One popular approach to decentralization is based on the use of Voronoi cells and the underlying Delaunay communication graph, where robots communicate with their immediate neighbors only, as seen in \cite{doi:10.1177/0278364908100177,TERUEL201951, Elamvazhuthi_Berman_2019}, with state-invariance recently proposed in \cite{GassuainVoronoiPathPlanning}. However, these types of controllers rely on precise position measurements of neighbors, which means they are not suitable in environments where localization is affected. Another decentralized approach is with Optimization-Based Controllers (OBCs) due to their fast in-the-loop control using the well-established Control Lyapunov Functions (CLFs) \cite{CLF_PDE_NoDensity}. 
More recently, due to its seamless integration with CLFs, Control Barrier Functions (CBFs) are used to ensure state invariance by enforcing forward invariance of a safe set, keeping robots within a safety boundary \cite{CBF}. 
However, in the field of Partial Differential Equations (PDEs), the application of CBFs remains underexplored. 

In this paper, we propose a novel decentralized OBC for density control of multi-robot systems with state invariance constraints suitable for dynamic and unpredictable environments. 
Our framework integrates both CLFs and CBFs within PDE-governed optimization problems, explicitly accounting for localization and motion noise. 
We derive sufficient conditions for theoretical guarantees of safety and convergence, and validate the controller in simulations and experiments with quadcopters, demonstrating its computational efficiency and practical applicability. 
To our knowledge, this is the first in-the-loop decentralized PDE swarm controller that explicitly incorporates localization and motion noise. Existing density controllers typically assume centralized coordination and perfect sensing, 
leaving no direct baseline for comparison.

\subsection{Related Work and Contributions}
PDE-based density control has gained popularity due to its improved scalability over ODE systems, as it describes the spatial density of the entire team rather than individual trajectories. In our previous work, \cite{RalPaper}, we incorporated both localization and motion noise in a decentralized OBC governed by a Fokker-Planck PDE. 
The effectiveness of the controller is shown with a Lyapunov-based proof, supported by simulations and experiments. In this paper, we further extend this formulation to enforce set invariance for safety.

Recently, mean-field game theory has been extended to optimal control with state invariance
\cite{Elamvazhuthi_Berman_2019, Optimal_withInvariance, OptimalMain}, establishing proofs of stability and optimality. Despite these advances, optimal control strategies face a common drawback: they are unresponsive to unexpected disturbances; any unmodeled forces, such as wind or animals, can invalidate the precomputed control. In contrast, our proposed controller operates within a real-time, decentralized framework that enables adjustments as disturbances arise. 
Another well-known approach to density control uses Voronoi cells, as seen in \cite{doi:10.1177/0278364908100177}, where a decentralized adaptive controller was developed for environmental sensing.
Recently, the authors of \cite{GassuainVoronoiPathPlanning}
developed a Gaussian distribution-based centroidal Voronoi tessellation approach for the density transport problem with known spatial obstacles. While effective, these algorithms require perfect localization, unlike our controller, which achieves safe control even under localization noise.

Recent advances in hardware have enabled fast in-the-loop optimization solvers, accelerating research in decentralized OBCs.
The authors of \cite{CLF_PDE_NoDensity} applied CLFs in PDE-governed robotic swarms for decentralized formation control, but did not incorporate safety constraints.
In the ODE domain, CBFs have become the standard framework for state invariance constraints \cite{CBF}, yet it remains underexplored in the PDE domain.
The work in \cite{boundary_control_cbf} recently employed a CLF-CBF OBC to maintain stability and safety in a traffic flow model through PDE-governed boundary control. In contrast, our work focuses on ensuring multi-robot safety in terms of obstacle avoidance, rather than regulating flow dynamics.

Finally, in \cite{CBF-PDE2ODE}, the authors achieved safe control of a robotic arm with CBF by transforming the PDE-governed structural deformation to ODEs using the Finite Element Method (FEM). While the authors of \cite{PDE-CBF-THermal} formulated a CBF-based back-stepping algorithm for the Stefan model, ensuring safety across both space and time with the Finite Difference Method (FDM). 
Similarly, our framework applies the FDM to obtain an ODE approximation of the Fokker-Planck PDE system. 
Given the smooth Gaussian PDFs and uniform spatial discretizations, the FDM introduces little discretization error and offers superior computational efficiency over FEM.

In this paper, we propose a decentralized OBC for multi-robot density control with state invariance governed by the Fokker-Planck PDE. Unlike prior approaches that assume perfect localization or neglect unmodeled disturbances, our controller offers a decentralized density control framework that adjusts to the current environment in real time, even with sensor noise. By integrating CBFs to ensure forward invariance, this controller enforces safety, improves robustness against disturbances, and remains computationally efficient. 

\section{Mathematical Background}
\label{sec: math background}
Our goal is to match the spatial density of \(N\) robots to a desired target distribution while enforcing state constraints for danger areas. Building on our previous work \cite{RalPaper}, with the averaging technique from \cite{Archer2004DynamicalDF, ANNUNZIATO2013487} and Itô calculus, we obtain the Fokker-Planck PDE governing the robot density dynamics with control fields \(u =(u_1,\dots,u_N)^T \in \mathcal{U}^N\):
\begin{equation}
    \frac{\partial \rho(r, t)}{\partial t} = \sum_{i \in N} \bigg[-\nabla_i(u_i \rho_i(r,t)) + T\Delta_i \rho_i(r, t) \bigg],
    \label{eq:FP_equation}
\end{equation}
where each \(u_i\) is a 2D vector characterizing a uniform control field for robot \(i \in \{1, \dots ,N\}\) within the admissible control set \(\mathcal{U} := \{v \in \mathbf{R}^2 : \|v\|_\infty \le u_{\max}\}\). The diffusion constant \(T\) is proportional to the maximum magnitude of uniform Gaussian motion noise \cite{RalPaper}. The total robot density \(\rho(r, t)\) at position \(r\) and time \(t\) is represented by a sum over \(N\) robots' Gaussian PDFs centered at robot positions \(x_i\):
\(
    \rho(r, t) = \sum_{i=1}^{N} \exp\left(-\frac{1}{2}(r - x_i)^T \Sigma (r - x_i)\right),
\)
where the covariance matrix \(\Sigma\) represents a robot's position uncertainty, assumed identical for all agents. Importantly, equation \eqref{eq:FP_equation} is deterministic, as it characterizes the swarm density distribution from a statistical perspective rather than the stochastic trajectories of individual robots.

With a model for the density of robots, we first formulate a centralized OBC. Then, in the next section, we present the decentralized controller with sufficient conditions for set-invariance guarantee. 
Developments in CBFs have provided a structural framework that is similar to the well-known CLFs. Naturally, an optimization-based framework with both CLF and CBFs has been studied and shown great success in \cite{CBF, wang2017safety}. Following the CLF-CBF optimization framework, we define the CLF as a measure of deviation from a possibly time-dependent target density \(\rho_d(r,t)\) in the spatial field \(\Omega \subset \mathbb{R}^2\). To ensure safety, a barrier function is designed to limit the distribution of robots within a danger area \(\mathcal{A} \subset \Omega\) by a threshold \(\epsilon > 0\). Specifically: 
\begin{equation}
    V(\rho) = \|\rho_d - \rho \|^2_{L^2(\Omega)}, \quad h(\rho) = \epsilon - \|\rho\|^2_{L^2(\mathcal{A})}.
    \label{eq: v_h def}
\end{equation}
With these measurements defined, the centralized CLF-CBF optimization-based control problem is formulated \cite{CBF}:
\begin{equation}
    \begin{aligned}
    \min_{u\in \mathcal{U}^N,\,s \in \mathbf{R}^N} &\quad \sum_{i=1}^N \big( \|u_i\|^2 + \gamma s_i^2 \big)\\
    \quad \text{s.t.} &\quad \alpha V(\rho) + \dot{V}(\rho,u) - s\leq 0\\
    &\quad \beta h(\rho) + \dot{h}(\rho,u) \geq 0,
\end{aligned}
\label{eq: Centralized Full model}
\end{equation}
where \(s_i, \gamma \geq 0\) are the slack variable and scaling factor, respectively, used to prioritize CBF over CLF, \(\alpha, \beta > 0\) are positive coefficients, and \(\dot V, \dot h\) are time derivatives of \eqref{eq: v_h def}. While this centralized controller ensures global safety, it requires real-time communication with all robots, which is often unrealistic in large fields or with a large team of robots. To address these limitations, we introduce a decentralized controller in the next section. 

\section{Decentralized Control}
A standard procedure towards the decentralization of safety constraints is to divide the scalar threshold \(\epsilon\) evenly between \(N\) robots:
\(
 h_i = \frac{\epsilon}{N} - \int_\mathcal{A} \rho_i \sum_{j=1}^N \rho_j dr.
\)
While this makes mathematical sense, \(\sum_i^N h_i = h\), enforcing this constraint requires every robot to have real-time global knowledge of all robots. This effectively implies that global communication and localization are available among the whole team, defeating the purpose of decentralization, as discussed above. 
In this section, we first propose a framework for a decentralized controller that does not require global communications, and then, a formal safety guarantee for the robot density with the controller is derived.

\subsection{Decentralization Formulation}
\label{subsec: Decentralize formulation}
In application, robots only have local communication capabilities with physically close neighbors within some distance \(D\). We define the neighbors of robot \(i\) with a distance-based communication graph, \(\Delta\)-disk graph \cite{mesbahi2010graph}, as:
\(\mathcal{N}_i = \left\{ j \in N\, \middle| \, \| {x}_j - {x}_i \| \leq D \right\}\), and \(\mathcal{N}_i^0 = \mathcal{N}_i \setminus \{i\}\).

Position data is passed through the communication links between each robot and its neighbors. Intuitively, if robot \(i\) knows its neighbors are far from danger, then \(i\) should be allowed to get closer to \(\mathcal{A}\) as the threshold \(\epsilon\) is shared globally. Instead of using a uniformed threshold of \(\frac{\epsilon}{N}\), robots can use a relaxed threshold \( N_i \frac{\epsilon}{N}\), where \(N_i = |\mathcal{N}_i|\) is the number of neighbors of robot \(i\). We define the decentralized safety barrier and Lyapunov function as:
\begin{equation}
    h_i = N_i \frac{\epsilon}{N} - \int_\mathcal{A} \rho_{\mathcal{N}_i}^2 dr, 
    \quad V_i = \int_\Omega (\rho_d - \rho_{\mathcal{N}_i})^2dr,
\label{eq: h_i v_i defination}
\end{equation}
where \(\rho_{\mathcal{N}_i} = \sum_{j\in{\mathcal{N}_i}}\rho_j\). 
Here, safety is calculated on a neighborhood level rather than an individual robot level, relaxing the safety threshold if the neighbors are far from \(\mathcal{A}\). Its relation to global safety \eqref{eq: v_h def} is shown in Proposition~\ref{prop: main}.

However, one crucial question arises: Is the safety of the entire team guaranteed using a decentralized OBC with \eqref{eq: h_i v_i defination}, which does not consider the motion of neighbors? The neighbors' optimized motion should not be communicated as it would result in a distributed optimization network, which requires more computing power and communication overhead \cite{distributed_opt_survey}. 
To include the motion of neighbors without explicit communication, we define a conservative worst-case prediction, \(\dot \rho_{j,i}^\mathbf{neighbor}\), for each neighbor \(j\) of robot \(i\). This worst-case behavior is one that causes the largest decrease in \(h_i\)---corresponding to the most unsafe action that robot \(i\)'s neighbors can take; thus, we can form the following optimization problem:
\begin{equation}
    \begin{aligned}
    \min_{u_{j,i} \in \mathcal{U}, \, j \in \mathcal{N}_i^0}  &\quad \beta_i h_i + \dot h_i^\mathbf{neighbor}
    \end{aligned}
    \label{eq: opt_step1}
\end{equation}
where \(\beta_i \geq 0\) as the individual CBF coefficients, with  
\[
\dot{h}_i^{\mathbf{neighbor}} = \sum_{j\in{\mathcal{N}_i^0}} \dot h_{j,i} = -2 \sum_{j\in{\mathcal{N}_i^0}} \int_\mathcal{A} \rho_{\mathcal{N}_i} \dot \rho_{j,i}^\mathbf{neighbor} dr
\]
as the prediction of all neighbors of \(i\) moving in the most unsafe direction. The time derivative of densities, \(\dot \rho_{j,i}\), follows \eqref{eq:FP_equation} with \(u_{j,i}\), robot \(j\)'s control predicted by \(i\), in the admissible set \(\mathcal{U}\).
Conversely, we can find the maximum increase in safety by robot \(i\), and define \(\dot h_i^\mathbf{self} = \max -2\int_\mathcal{A} \rho_{\mathcal{N}_i} \dot \rho_{i}dr\).

Notably, there are edge cases where the safest control by robot \(i\) cannot neutralize the most aggressive movement of its neighbors. To quantify this, we define a variable \(\delta_i = \min \{0, \lambda_i + \dot h_i^\mathbf{self}\}\), with \(\lambda_i\) as the objective value of \eqref{eq: opt_step1}. The relaxation variable \(\delta_i\) ensures the feasibility of each CBF. This is used in the main optimization problem:
\begin{equation}
    \begin{aligned}
    \min_{u_i\in \mathcal{U},\,s_i\in \mathbf{R}} &\quad \|u_i\|^2 + \gamma s_i^2\\
     \quad \text{s.t.} &\quad \alpha_i{V}_i + \dot{V}_{i} - s_i\leq 0\\
    &\quad \beta_i {h}_{i} + \dot{h}_{i} \geq \delta_i,
    \end{aligned}
    \label{eq: opt_main}
\end{equation}
where \(\alpha_i \geq 0\) is the CLF weight.
If robot \(i\) cannot counteract all neighbors' most unsafe movement, \(\delta_i < 0\) will relax the CBF to ensure feasibility with the safest command. Otherwise, \(\delta_i = 0\) is the standard CBF allowing robot \(i\) to get closer to \(\mathcal{A}\) to match the target PDF, \(\rho_d\). 
This relaxation encodes the worst-case assumption for all neighbors to account for limited information in practice. If all robots act accordingly, assuming the most unsafe neighbor behaviors, then the swarm remains collectively safe. The flowchart of the proposed framework can be seen in Fig.~\ref{fig: flow diagram}.
\begin{figure}
    \centering
    \includegraphics[width=\linewidth]{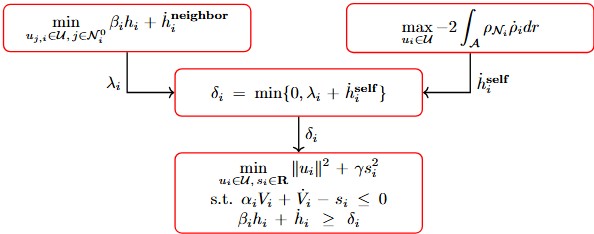}
    \caption{Schematic for the safe decentralized OBC. Each robot \(i\) computes the worst-case safety threat, \(\lambda_i\), from its neighbors, and the maximum safety improvement, \(\dot h_i^\mathbf{self}\), by itself. These parameters determine if a slack \(\delta_i\) is necessary for the CBF feasibility in controller \eqref{eq: opt_main}. }
    \label{fig: flow diagram}
\end{figure}

\subsection{Safety Guarantee}
\label{sec: proof} 
In this section, we show the sufficient conditions such that the sum of the decentralized CBF, \(\beta_i {h}_{i} + \dot{h}_{i} \geq \delta_i\), in \eqref{eq: opt_main} guarantees the centralized global CBF, \(\beta h + \dot{h} \geq 0\), in \eqref{eq: Centralized Full model}.

\begin{proposition}[Sufficient Condition for Local-to-Global Safety Guarantee]
\label{prop: main}
    Let each robot \(i \in \{1, \dots ,N\}\) implement the proposed controller \eqref{eq: opt_main}. Then, if every robot satisfies the bound below (e.g., by tuning \(\beta, \beta_i, D, \epsilon \dots\))
    \begin{equation}
      2 \int_\mathcal{A} \dot\rho_i  \rho_{\bar{\mathcal{N}}_i} \ dr \leq \beta \bigg( \frac{\epsilon}{N} - \int_\mathcal{A}\rho_i \big( \rho_{\mathcal{N}_i} + \rho_{\bar{\mathcal{N}}_i} \big) dr \bigg) + \dot h_i^\mathbf{self},
      \label{eq: final bound}
    \end{equation}
    then the global CBF condition,
    \(
    \beta h + \dot h \geq 0,
    \)
    is guaranteed to hold. 
\end{proposition}

\begin{proof}
    We know \(\dot h_i \geq \delta_i- \beta_i h_i\) is the solution of \eqref{eq: opt_main}, which implies \(\sum_{i=1}^{N} \dot h_i \geq \sum_{i=1}^{N} ( \delta_i - \beta_i h_i)\). To show that \(\dot h \geq -\beta h\) in \eqref{eq: Centralized Full model} by \eqref{eq: opt_main}, the below equation is sufficient:
    \begin{equation}
    \dot h - \sum_{i=1}^{N} \dot h_i \geq -\beta h - \sum_{i=1}^{N} ( \delta_i - \beta_i h_i).
    \label{eq: ineq goal}
    \end{equation}
By definition: \(
    \dot h =-2 \sum_{i=1}^{N} \int_\mathcal{A} \big(\rho_{\mathcal{N}_i} + \rho_{\bar{\mathcal{N}}_i} \big)\dot\rho_i\ dr,
    \)
    where \(\bar{\mathcal{N}}_i\) is all robots not in the detection radius, \(D\), of \(i\). By \eqref{eq: opt_step1}, we have: \( 
    \dot h_i = -2 \int_\mathcal{A} \rho_{\mathcal{N}_i} \big(\dot \rho_i + \sum_{j \in \mathcal{N}_i^0} \dot\rho_{j,i}^\mathbf{neighbor} \big)dr \, 
    = -2\int_\mathcal{A} \rho_{\mathcal{N}_i} \dot \rho_i \ dr + \lambda_i - \beta_i h_i.
    \)
This simplifies the LHS of \eqref{eq: ineq goal}:
\begin{equation}
        \dot h - \sum_{i=1}^{N} \dot h_i\ = 
         -2 \sum_{i=1}^{N} \int_\mathcal{A} \dot\rho_i  \rho_{\bar{\mathcal{N}}_i} \ dr - \sum_{i=1}^{N} \bigg( \lambda_i - \beta_i h_i \bigg),
\end{equation}
which means \eqref{eq: ineq goal} is equivalent to:
\begin{equation}
    \sum_{i=1}^{N} \bigg[ 2\int_\mathcal{A} \dot\rho_i  \rho_{\bar{\mathcal{N}}_i} \ dr + \lambda_i \bigg]
     \leq \beta h + \sum_{i=1}^{N} \delta_i.
\label{eq: ineq goal expanded}
\end{equation}

From the definition of \(\delta_i\), we know either \(\delta_i  = \lambda_i + \dot h_i^\mathbf{self} \), or \(\lambda_i \geq - \dot h_i^\mathbf{self}\), when \(\delta_i = 0\).
The lower bounds in both cases are the same, giving a swarm-level sufficient bound that guarantees \eqref{eq: ineq goal expanded}, which is equivalent to \eqref{eq: ineq goal}, implying the global CBF condition we want to show:
\begin{equation}
    2\sum_{i=1}^{N} \int_\mathcal{A} \dot\rho_i  \rho_{\bar{\mathcal{N}}_i} \ dr \leq \beta h + \sum_{i=1}^{N} \dot h_i^\mathbf{self}.
    \label{eq: final bound swarm}
\end{equation}

To obtain an individual level bound, \(h\) can be divided as 
\(
h = \sum_{i=1}^{N} \big( \frac{\epsilon}{N} - \int_\mathcal{A}\rho_i ( \rho_{\mathcal{N}_i} + \rho_{\bar{\mathcal{N}}_i} ) dr \big),
\)
which leads to a tighter individual robot bound \eqref{eq: final bound}. 
\end{proof}

\begin{remark}
Proposition~\ref {prop: main} establishes a sufficient condition for global safety by the proposed controller. Furthermore, under the assumption that the system starts in a safe state, \(h \geq 0\), this condition ensures forward invariance of the set \(\{\rho: h( \rho)\ge0\}\). While it is impractical for every robot to compute this bound in real-time, as it requires knowledge of non-neighbors, this bound can be used for simulation-based tuning before deployment and offline verification afterward.     
\end{remark}

\begin{remark}
\label{remark simple bound}
Since the swarm starts in a safe state, \(h \ge 0\), the swarm-level safety condition, \eqref{eq: final bound swarm}, simplifies to:
    \begin{equation}
\label{eq:avg-L2-condition}
\sum_{i=1}^N \big\|\rho_{\bar{\mathcal N}_i}\big\|_{L^2(\mathcal A)}
\;\le\; 
\sum_{i=1}^N \big\|\rho_{\mathcal N_i}\big\|_{L^2(\mathcal A)},
\end{equation}
by applying the Cauchy–Schwarz inequality to both sides, noting that the RHS holds as equality by construction. Then, a sufficient condition for safety is that, on average, neighbors remain closer to the danger zone than non-neighbors.
\end{remark}

The individual bound \eqref{eq: final bound} is also intuitive. 
The LHS quantifies the risk of the commanded control on non-neighbors, \(\rho_{\bar{\mathcal{N}}_i}\). The RHS sets the upper bound with the current safety level and potential improvements, incorporating the global CBF coefficient \(\beta\), safety threshold \(\epsilon\), robot's spatial configuration, and its achievable improvement in safety \(\dot h_i^\mathbf{self}\). While the individual \(\beta_i\) isn't explicit, it appears in \(\delta_i\), which leads to a more relaxed bound than \(\dot h_i^\mathbf{self}\). Clearly, higher \(\beta, \epsilon, D\), or lower \(\beta_i\) all relax the inequality \eqref{eq: final bound}. Additionally, the sparsity of the detection-based communication graph also plays a critical role, since a larger detection radius reduces \(\rho_{\bar{\mathcal{N}}_i}\) in Remark~\ref{remark simple bound}. At the extreme, when detection covers the whole field, the summation bound \eqref{eq: final bound swarm} is always satisfied, resembling a centralized approach. 
However, the controller may fail near complex boundaries of \(\mathcal{A}\). If a robot is closely surrounded by \(\mathcal{A}\), i.e., trapped in a small ring of fire, diffusion alone can violate the safety constraint even in a centralized setting. In such situations, no control is safe under the current Fokker-Planck equation formulation with constant diffusion \eqref{eq:FP_equation}. This observation highlights the importance of the spatial structure of \(\mathcal{A}\) and motivates the following remark. 

\begin{remark}
    As stated in Section~\ref{sec: math background}, the Fokker-Planck equation \eqref{eq:FP_equation} contains a diffusion coefficient \(T\) representing the random motion noise of a robot. We define a constant \(T = 0.3^2 u_\mathbf{max}/ 2 = 0.045 u_\mathbf{max}\), which ensures that \(99\%\) of motion noise lies in \([-u_\mathbf{max}, u_\mathbf{max}]\). A variable \(T(u) = 0.045 u\) can more accurately represent the command-dependent motion noise. We decided to use a constant maximum noise for two reasons: it is computationally cheaper and captures other forces, such as inertia in a quadcopter, which persists even if the commanded motion is zero. 
\end{remark}

\begin{figure*}
    \subfloat{\includegraphics[width=0.165\linewidth]{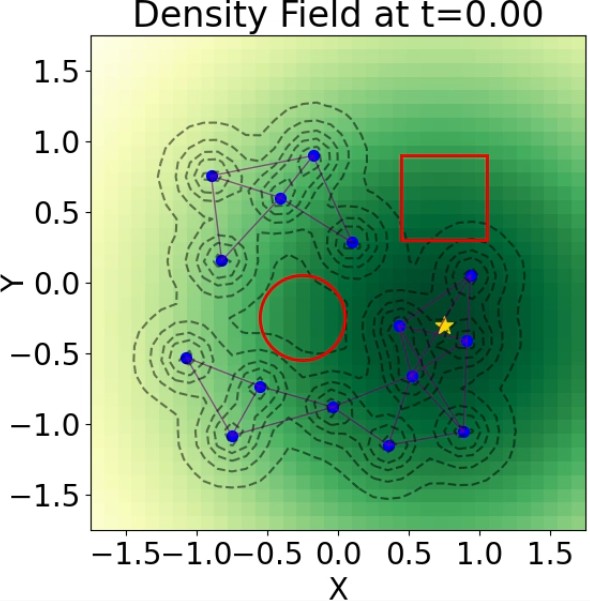}}%
    \subfloat{\includegraphics[width=0.165\linewidth]{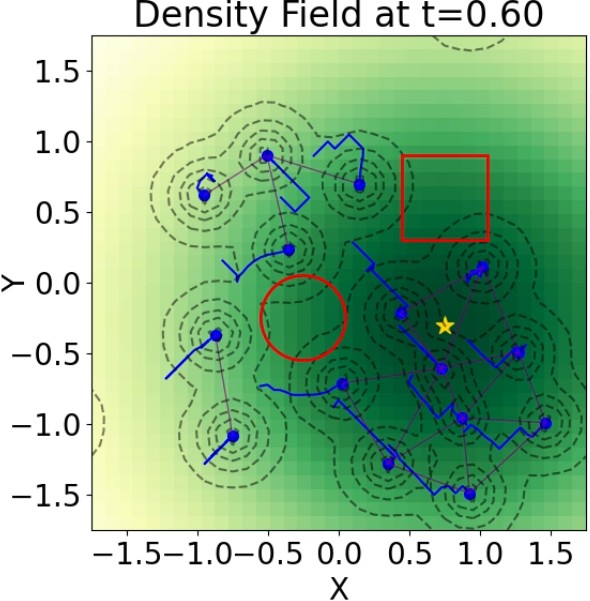}}%
    \subfloat{\includegraphics[width=0.165\linewidth]{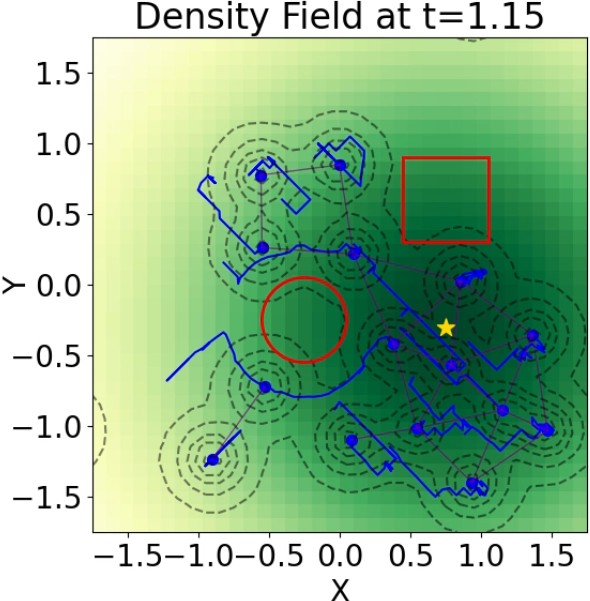}}%
    \subfloat{\includegraphics[width=0.165\linewidth]{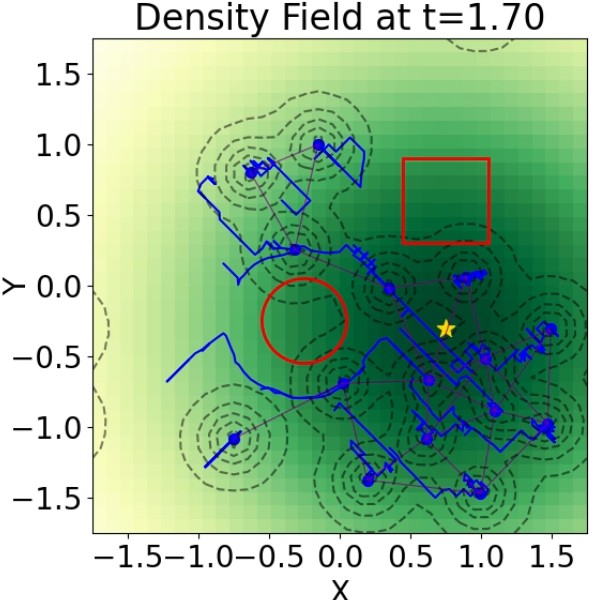}}%
    \subfloat{\includegraphics[width=0.165\linewidth]{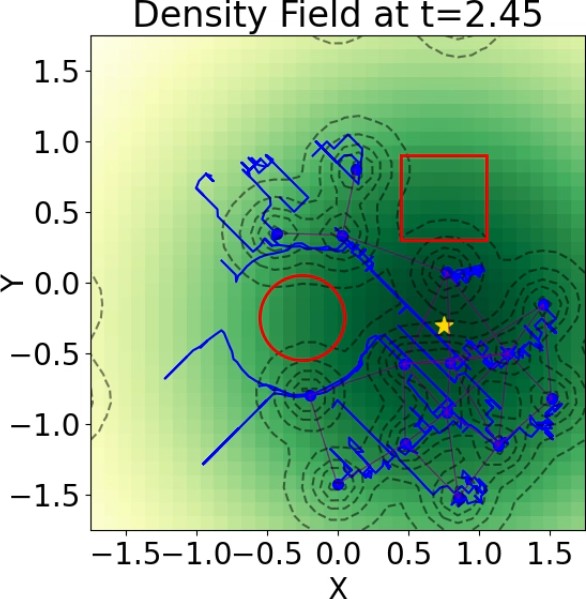}}%
    \subfloat{\includegraphics[width=0.165\linewidth]{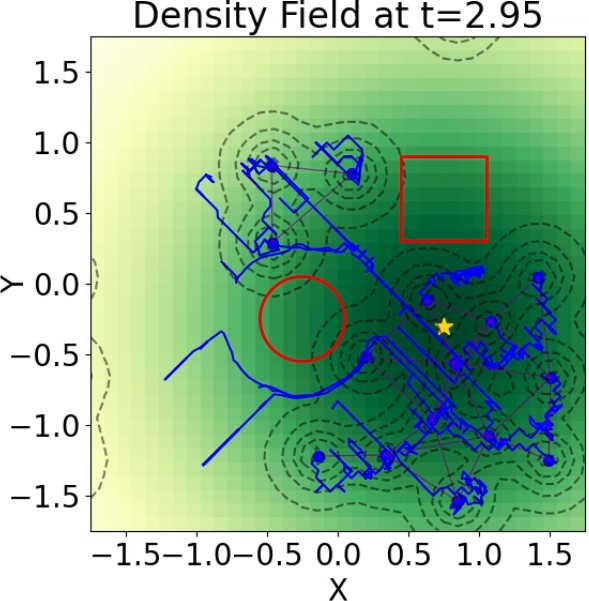}}
    \caption{Time sequence of the simulation. The star indicates the center of the target, green is the target density (darker is higher), and red is the invariance area. Blue represents robots' positions with trajectory trails, purple represents communication links, and the dashed contour is the robot's PDF.}
    \label{fig:Swarm_Time_Sequence}
\end{figure*}

\subsection{Numerical Implementation}
To implement the controller, we approximate the infinite-dimensional PDEs in \eqref{eq: opt_main} as ODEs using a central FDM. On a uniform \(N_x \times N_y\) spatial field with grid size \(l\), the infinite-dimensional \(\rho_i, \rho_d\in L^2(\Omega)\) are approximated and flattened to 1D as \(\hat{\rho_i}, \hat{\rho}_d \in \mathbf{R}^{N_d}\) with \(N_d = N_x N_y\). Then, the optimization problem \eqref{eq: opt_main} becomes:
\begin{equation}
    \begin{aligned}
    \min_{\substack{u_i \in \mathcal{U} , s_i \in \mathbf{R}}}
    &\quad \|u_i\|^2 + \gamma s_i^2\\
    \text{s.t.} \quad& l^2 \sum_{k=1}^{N_d} \bigg[\big(\hat{\rho}_d^k - \hat{\rho}_{\mathcal{N}_i}^k\big)^2 -2 \big(\hat{\rho}_d^k - \hat{\rho}_{\mathcal{N}_i}^k\big)\hat{\rho}_{i}^{t,k}\bigg] \leq s_i
    \\
    & \beta_i\big[N_i \frac{\epsilon}{N} - l^2 \sum_{k \in \mathcal{A}}(\hat{\rho}_{\mathcal{N}_i}^k)^2\big] \\
    &  - 2 l^2 \sum_{k \in \mathcal{A}} \big( \hat{\rho}_{\mathcal{N}_i}^k \big(
    \hat{\rho}_{i}^{t,k} +
    \sum_{j \in \mathcal{N}_i^0} \hat{\rho}_{j,i}^{t,\mathbf{neighbor}, k}\big)\big) \geq \delta_i,
    \end{aligned}
    \label{eq: opt_main_discritized}
\end{equation}
where the superscript \(k\) denotes the 1D grid index, the flattened representation of the 2D field in code. The superscript \(t\) is the approximation of the Fokker-Planck equation \eqref{eq:FP_equation}:
\[
\hat{\rho}_{j,i}^t = -(u_{j,i}^x A_x + u_{j,i}^y A_y)\hat{\rho}_j + T \, B \hat{\rho}_j.
\]
The constant sparse matrices \(A_x, A_y, B  \in \mathbf{R}^{N_d \times N_d}\) are derived from central difference scheme corresponding to\(\nabla_x \rho, \nabla_y \rho\) and \(\Delta \rho\). The control input \(u_{j,i}^x, u_{j,i}^y\) denotes robot \(i\)'s prediction of robot \(j\)'s velocity in the \(x,y\) directions, respectively. Thanks to the ODE approximation, \eqref{eq: opt_main_discritized} is a strictly convex optimization problem with a quadratic objective and two linear constraints in \(u_i\), allowing for rapid computation in real-life applications.

\begin{proposition}[Existence, Uniqueness, and Safety of the Discretized Decentralized Controller]
    Following the proposed framework, the constrained optimization problem \eqref{eq: opt_main_discritized} is always feasible. Under Proposition~\ref{prop: main}, its unique solution \(\hat{u}_i^*\) guarantees the forward invariance of the set $\{\hat \rho_i : \hat h \geq 0\}$, where \(\hat h\) is discretized safety barrier \(h\).
\end{proposition}
\begin{proof}
    By construction, following the framework in Fig.~\ref{fig: flow diagram}, the CBF constraint is always feasible with at least one \(\hat u_i\) that yields \(\hat h_{i}^{t,\mathbf{self}}\). The CLF constraint, the first constraint in \eqref{eq: opt_main_discritized}, is always satisfied by relaxing the slack variable \(s_i\). Therefore, \eqref{eq: opt_main_discritized} is always feasible and satisfies the CBF constraint to ensure forward invariance of the set $\{\hat \rho_i: \hat h \geq 0\}$ by the bounds \eqref{eq: final bound} or \eqref{eq: final bound swarm}. Finally, as the objective function is strictly convex, the optimal solution \(\hat{u}_i^*\) is unique.
\end{proof}

\begin{remark}
    It is important to remember that \eqref{eq: opt_main_discritized} involves competing stability and safety constraints, such that $\hat\rho_i$ might converge to a local minimum. Thanks to the discretization proposed in \eqref{eq: opt_main_discritized}, the ODE techniques developed in \cite{reis2020control} can be used to characterize the local minima.
\end{remark}

\section{Application and Results}
With a theoretical guarantee for convergence to at least local minima under the state invariance constraint, in this section, the proposed controller is tested in simulations and quadcopter experiments using 4 Seeker quadcopters \cite{seeker_datasheet}. 
While the formulation generalizes to 3D naturally, the tests are conducted in 2D, as many coverage applications focus on 2D manifolds, such as wildfire suppression and agricultural field monitoring. Moreover, inter-robot collision is not modeled explicitly, as avoidance is assumed at the robot level. Additionally, the density matching cost, \(V(\rho)\), discourages collisions by penalizing proximity.

\subsection{Simulation}
The simulation is set up with a team of 15 robots in a $3.5 \times 3.5$~m environment with periodic boundary conditions. The space is discretized into $10 \times 10$~cm squares. The simulations run for $3$ seconds with $0.05$~s timesteps. The robots all have a detection radius of $1$~m and share the same hyperparameter. The green target PDF is centered at \((0.75, -0.3)\), and the two danger areas, denoted by \(\mathcal{A}\), are highlighted in red as shown in Fig.~\ref{fig:Swarm_Time_Sequence}. Note that \(\mathcal{A}\) represents areas, not single obstacles; multiple nearby obstacles can be grouped into one \(\mathcal{A}\) for computational efficiency. 
While the system is designed to handle measurement and motion noise, we first look at a noise-free simulation to study the detailed behavior of the controller. Then, the controller is validated on 4 quadcopters with noise in Section~\ref{subsection: quadcopter}.
\begin{figure}
    \subfloat{\includegraphics[height = 3.7cm]{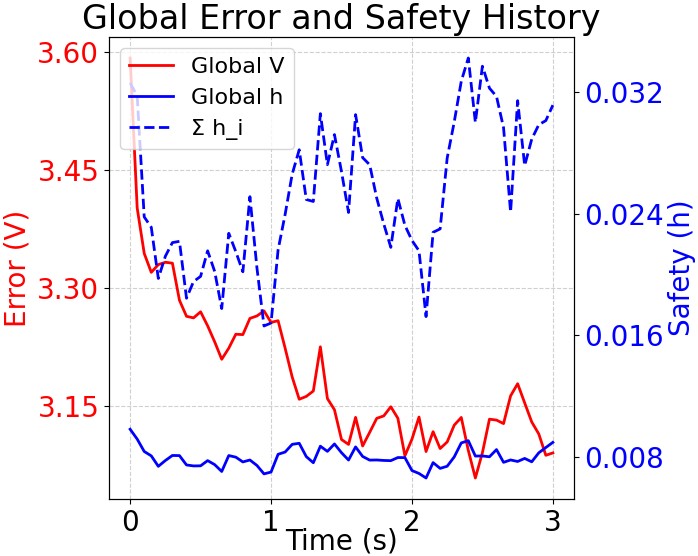}}\hfill
    \subfloat{\includegraphics[height = 3.7cm]{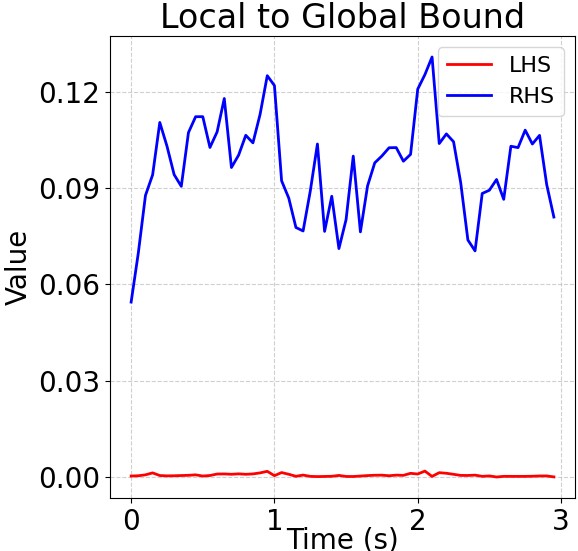}}
    \caption{Density matching, Safety, and Bound Plots. Lower \(V\) value represents better target tracking. Higher \(h\) corresponds to being further from the danger area. The right plot shows the terms in bound \eqref{eq: final bound swarm}.}
    \label{fig: Simulation plot}
\end{figure}

The full controller, as shown in Fig.~\ref{fig: flow diagram} and \eqref{eq: opt_main_discritized}, is implemented in Python. Simulations were performed on an AMD Ryzen 5 5600 6-core processor with Python API Mosek \cite{mosek}. Thanks to the convexity of \eqref{eq: opt_main_discritized}, the average computation time is \(0.003\)~s per robot per timestep, allowing for real-time deployment. 
The result is presented in Fig.~\ref{fig:Swarm_Time_Sequence}. 
As the robots move towards the target, in both the top and bottom left areas, the closest robots move towards the target while their neighbors back off to ensure safety. The next closest robots only start moving towards the target once the first robots exit their detection radii, lowering the local safety threat. Interestingly, in the top left corner of the last two plots, the closest robot moves away from the target instead of passing between the danger zones like the previous robots. This is due to its two neighbors being much closer to \(\mathcal{A}\) than before, and it has no other robots in the detection radius to offset this safety threat. Eventually, most robots converge to the target \(\rho_d\) while avoiding unsafe regions.

The left plot of Fig.~\ref{fig: Simulation plot} shows the global density matching error, \(V(\rho)\), global safety barrier function, \(h(\rho)\), and local safety sum, \(\sum_i^N h_i\), over time. As robots move closer to the danger area, \(h\) decreases but never falls below zero, demonstrating the proposed controller's safety guarantees. 
The local safety sum \(\sum_i^N h_i\) is higher since neighbors not nearby \(\mathcal{A}\) create extra safety padding. Sudden fluctuations are observed when new robots near \(\mathcal{A}\) enter the detection radius, triggering stricter safety requirements. This explains why the third robot, in the top left corner, cannot cross the danger zone. It detects more robots nearby \(\mathcal{A}\) as it tries to pass, but must retreat for safety. The right plot demonstrates the validity of the bound \eqref{eq: final bound swarm}, confirming the guarantee for the global safety CBF. As expected, the inequality is easily satisfied throughout the simulation. The edge case that might violate this bound is with a very small detection radius, and all robots simultaneously move towards \(\mathcal{A}\) at close distances, which is unlikely in practice. 
These results demonstrate that the proposed decentralized controller can maintain set-invariance for safety while matching a target PDF with multiple robots and danger zones.

\subsection{Quadcopter Experiment}
\label{subsection: quadcopter}
To test the performance of the controller in practice, experiments with 4 Seeker quadcopters are done in a $3 \times 3$~m cage. 
Notice that this setup results in a low-density swarm. Although our proposed approach does not require a high-density swarm, performance is expected to increase with a higher number of robots.
The target is centered at \((0.2, 0.6)\) with a circular danger area of radius $0.4$~m at \((-0.4, -0.3)\). Four quadcopters 
have detection radii set to $1$~m and collision radius to $0.7$~m, as a safety buffer beyond the physical tip-to-tip distance of $0.5$~m. The setup can be seen in Fig.~\ref{fig: exp_setup}. The items on the floor provide quadcopters with localization features for onboard cameras. Each quadcopter computes its own controller and only has the location information of its neighbors to mimic deployment in the wild. 

\begin{figure}
    \centering
    \subfloat[Experiment setup]{\includegraphics[width=\linewidth]{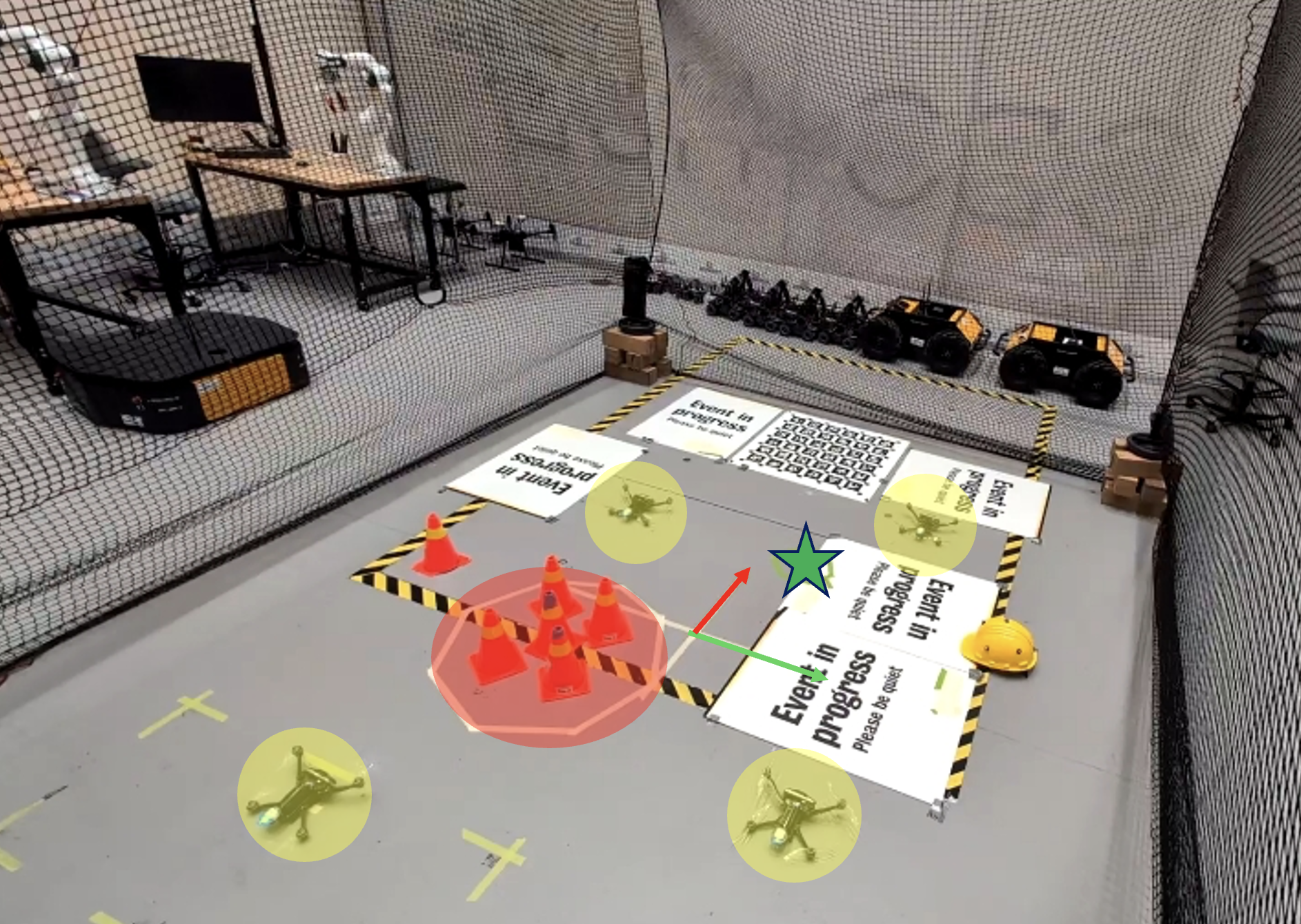}
    \label{fig: exp_setup}}\\    
    \subfloat[Time sequence]{
        \begin{minipage}{\linewidth}
        \centering
        \includegraphics[width=0.335\linewidth]{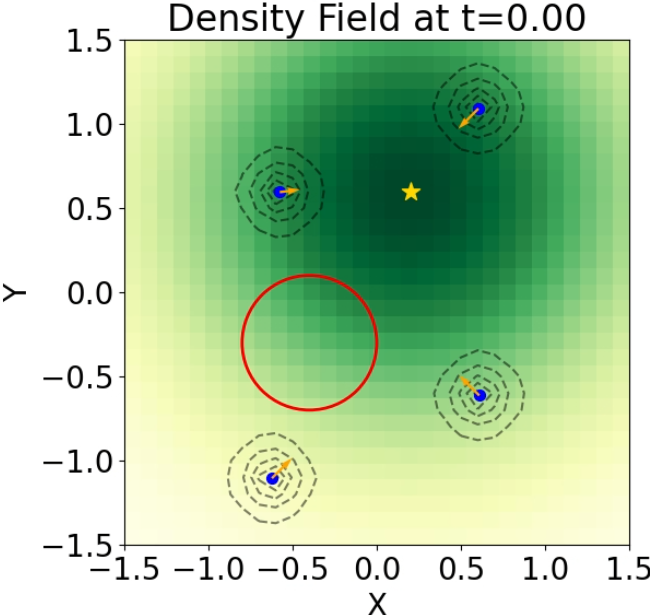}%
        \includegraphics[width=0.335\linewidth]{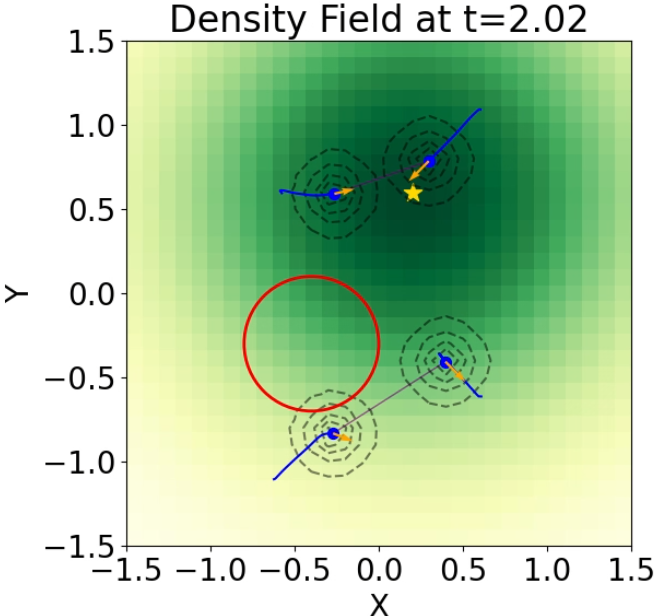}%
        \includegraphics[width=0.335\linewidth]{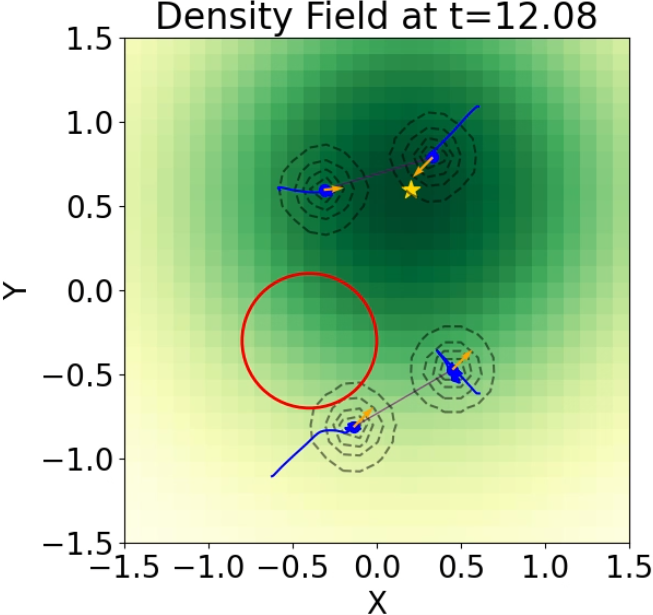}\\
        \includegraphics[width=0.335\linewidth]{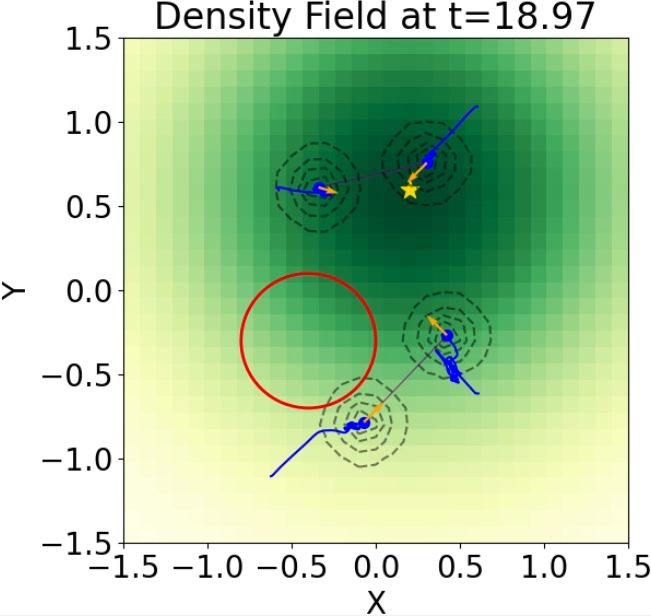}%
        \includegraphics[width=0.335\linewidth]{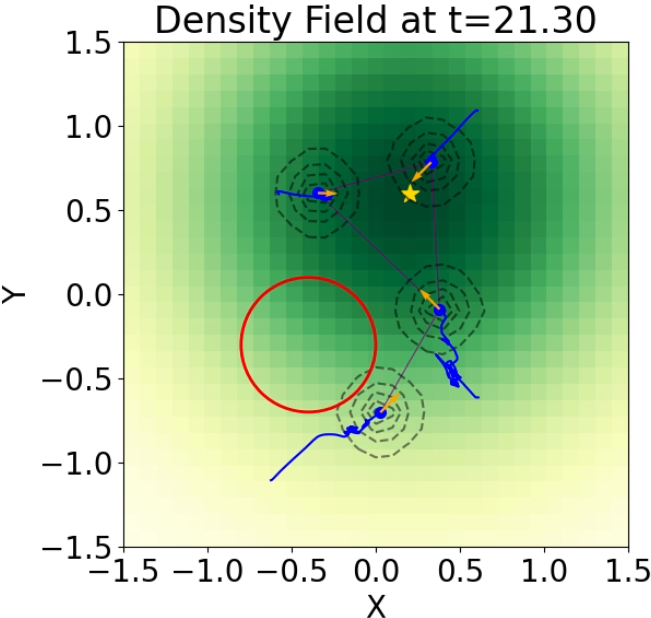}%
        \includegraphics[width=0.335\linewidth]{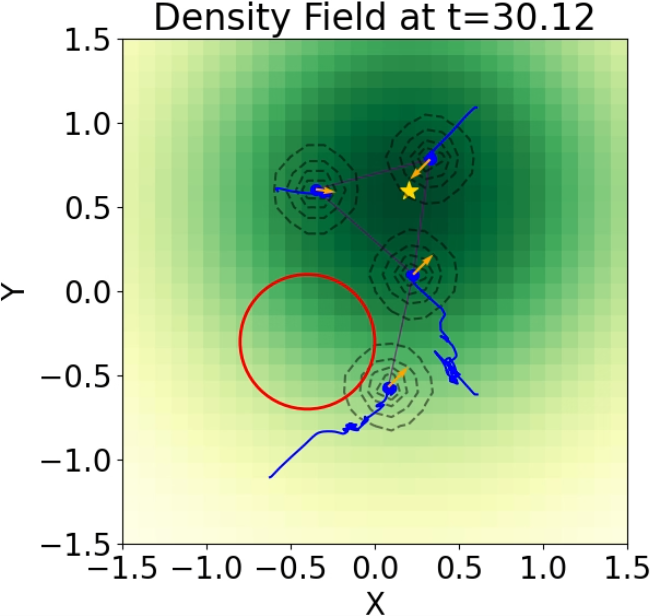}
        \end{minipage}
        \label{fig: quadcopter_Time_Sequence}
        }
    \caption{(a) Experiment setup. Quadcopters in yellow, \(\mathcal{A}\) in red, green X-axis, red Y-axis, and the green star is the center of \(\rho_d\).  (b) Time sequence of the experiment. The star indicates the center of the target, green is the target density (darker is higher), and red is the invariance area. Blue represents robots' positions with trajectory trails, purple communication links, yellow control commands, and the dashed contour is the robot's PDF.}
    \label{fig:combined_experiment}
\end{figure}

The time sequence of the experiment is shown in Fig.~\ref{fig: quadcopter_Time_Sequence}. The top two quadcopters had obstacle-free paths to the target and were able to converge directly, only to stop for collision avoidance. In contrast, the bottom two quadcopters exhibit more interesting behaviors. The right quadcopter's optimal path starts by curving towards the danger area, then straight north towards the target. However, as the left quadcopter enters its detection radius, the right quadcopter backs off because the left quadcopter is already close to the boundary of \(\mathcal{A}\). The left quadcopter, at this point, tries to maintain its path around the danger area, since its next step moves away from \(\mathcal{A}\). However, it stopped due to collision avoidance. 
Eventually, after some oscillations and thanks to the motion noise from air turbulence generated by other quadcopters, the right quadcopter shifted positions and found another locally optimal path to move towards the target. Meanwhile, the left quadcopter continues on its path around \(\mathcal{A}\) since the collision risk disappears with the right quadcopter's repositioning. 

\begin{figure}
    \centering
    \subfloat{\includegraphics[height = 3.7cm]{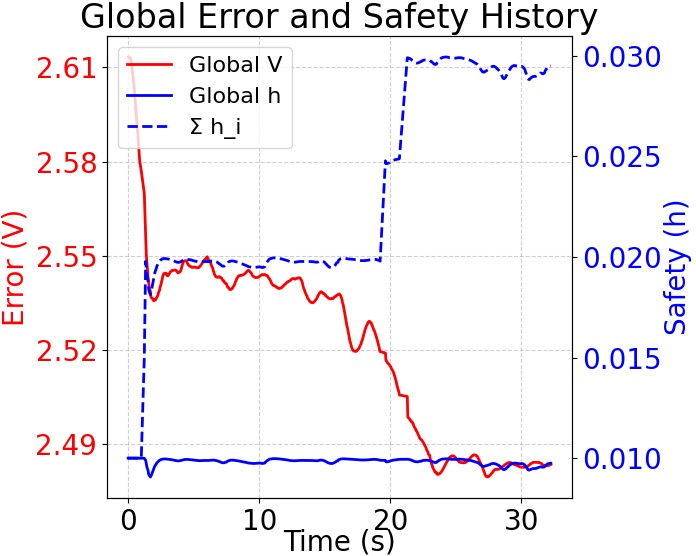}}\hfill
    \subfloat{\includegraphics[height = 3.7cm]{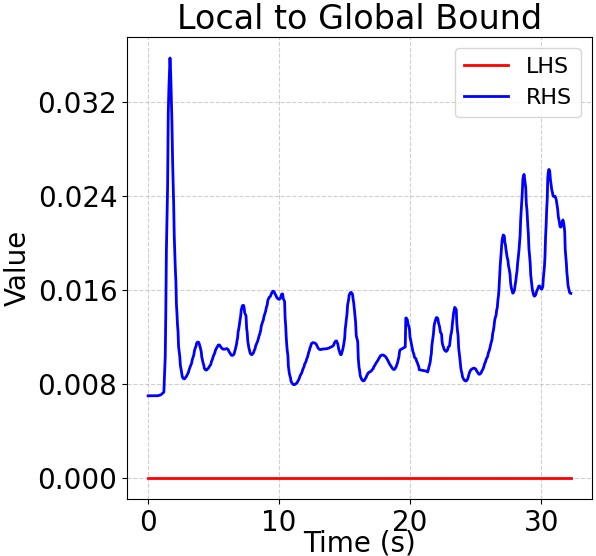}}
    \caption{Density matching, Safety, and Bound Plots. Lower \(V\) value represents better target tracking. Higher \(h\) corresponds to being further from the danger area. The right plot shows the terms in bound \eqref{eq: final bound swarm}.}
    \label{fig: quadcopter results}
\end{figure}
The corresponding metrics are plotted in Fig.~\ref{fig: quadcopter results}. The global safety value drops initially as the bottom left quadcopter moves towards \(\mathcal{A}\), but quickly recovers as both bottom quadcopters act safely.
When a robot, far from \(\mathcal{A}\), is detected by robot \(i\), the local safety \(h_i\) jumps.
The density matching error drops quickly at the beginning as the top two quadcopters converge. \(V(\rho)\) plateaus as the bottom two quadcopters are blocked by the controller and collision avoidance. Finally, \(V(\rho)\) reaches a minimum as the bottom quadcopters converge to the target. 
The inequality bound plot shows similar results as the simulation, but with a smaller magnitude difference due to fewer robots. This supports the claim that the controller is safer with more neighbors.

The full video of the experiment is attached and also available at \url{https://youtu.be/fWqWTY1BVX8}. The video shows that all quadcopters are influenced by heavy motion noise due to air turbulence from other quadcopters. This strong wind vibration also causes localization difficulties with the onboard cameras. Despite these challenges of noise, the controller was still able to guide the team toward the target PDF while ensuring global safety with only position data from neighbors in a limited detection radius. 

\section{Conclusion}
This paper presented a novel decentralized density controller for multi-robot systems that enforces set invariance to achieve robot safety. We outlined the framework for implementation and derived sufficient conditions for the decentralized safety constraints to guarantee global safety. The effectiveness of the OBC was demonstrated in both simulations and quadcopter experiments, despite measurement and motion noise. We envision this controller as an important step towards deploying teams of robots for safety-critical density coverage tasks, such as forest firefighting, pesticide deployment, terrain exploration, and search and rescue.

\bibliographystyle{IEEEtran}
\bibliography{reference}
\end{document}